\documentclass[letterpaper, 10 pt, conference]{ieeeconf}  

\IEEEoverridecommandlockouts                              
\usepackage{float,amsmath,amsfonts,amssymb,amsthm,color,graphicx}
\usepackage{graphics} 
\usepackage{booktabs} 
\usepackage{caption, subcaption}
\usepackage{breqn}
\usepackage{array, mathtools}
\usepackage{pgfplots}
\usepackage{siunitx}
\usepackage{algorithm}
\usepackage[noend]{algpseudocode}
\usepackage{tabulary}
\usepackage{stmaryrd} 
\usepackage{cite}

\usepackage{tikz}
\usepackage{pgfplots}
\usetikzlibrary{decorations.markings}
\usetikzlibrary{patterns}
\usetikzlibrary{arrows}

\usetikzlibrary{shapes}
\usetikzlibrary{fit}
 \usetikzlibrary{calc}
\usetikzlibrary{decorations.pathreplacing}
\usetikzlibrary{arrows,positioning} 
\usetikzlibrary{pgfplots.groupplots}

\usepackage{thmtools}
\declaretheorem[style=definition,name=Definition,qed=$\blacksquare$]{definition}
\declaretheorem[style=definition,name=Example,qed=$\blacksquare$]{example}

\declaretheorem[style=definition,name=Special Case,qed=$\blacksquare$]{specialblock}
\declaretheorem[style=definition,name=Remark,qed=$\blacksquare$]{remark}
\declaretheorem[style=definition,name=Proposition,qed=$\blacksquare$]{proposition}
\declaretheorem[style=plain,name=Theorem,qed=$\blacksquare$]{theorem}

\newcommand{\R}{\mathbb{R}}

\newcommand{\eR}{\overline{\R}}

\newcommand{\X}{\mathcal{X}}

\newcommand{\T}{\mathcal{T}}
\newcommand{\W}{\mathcal{W}}

\newcommand{\rect}[1]{\llbracket #1 \rrbracket}
\newcommand{\change}[1]{{\color{black} #1}}

\title{\LARGE {\bf Tight Decomposition Functions for Continuous-Time Mixed-Monotone Systems with Disturbances}}

\author{Matthew Abate, Maxence Dutreix, and Samuel Coogan
\thanks{This work was partially supported by the Air Force Office of Scientific Research under grant FA9550-19-1-0015.}
\thanks{
\change{
M. Abate, M. Dutreix, and S. Coogan are with the School of Electrical and Computer Engineering, Georgia Institute of Technology, Atlanta, 30332, USA {\tt\small $\{$Matt.Abate,\, Maxdutreix,\, Sam.Coogan$\}$ @GaTech.edu}. M. Abate is also with the School of Mechanical Engineering and S. Coogan is also with the School of Civil and Environmental Engineering.}}
}

\begin{document}

\maketitle
\thispagestyle{empty}
\pagestyle{empty}

\begin{abstract}
The vector field of a mixed-monotone system is decomposable via a decomposition function into increasing (cooperative) and decreasing (competitive) components, and this decomposition allows for, \emph{e.g.}, efficient computation of reachable sets and forward invariant sets.
A main challenge in this approach, however, is identifying an appropriate decomposition function.
In this work, we show that any continuous-time dynamical system with a Lipschitz continuous vector field is mixed-monotone, 
\change{
and we provide a construction for the decomposition function that yields the tightest approximation of reachable sets when used with the standard tools for mixed-monotone systems.
}
Our construction is similar to that recently proposed by Yang and Ozay for computing decomposition functions of discrete-time systems \cite{ozaything} where we make appropriate modifications for the continuous-time setting and also extend to the case with unknown disturbance inputs. As in \cite{ozaything}, our decomposition function construction requires solving an optimization problem for each point in the state-space; however, we demonstrate through example how tight decomposition functions can \change{sometimes} be calculated in closed form.  As a second contribution, we show how under-approximations of reachable sets can be efficiently computed via the mixed-monotonicity property by considering the \change{backward-time} dynamics. 
\end{abstract}

\section{Introduction}
Mixed-monotone systems are characterized by vector fields that are decomposable into increasing (cooperative) and decreasing (competitive) interactions. This allows for embedding the system dynamics into a higher dimensional system with twice as many states but for which the dynamics are monotone \cite{nonmonotone, kulenovic2006global, gouze1994monotone}\change{; an approach that is similar in spirit is first pioneered in \cite{smale1976differential}.}  Thus, decomposing the system dynamics enables one to apply the powerful theory of monotone dynamical systems to the higher dimensional \emph{embedding} system to conclude properties of the original system. For example, mixed-monotonicity allows for: efficiently approximating reachable sets by evaluating only one trajectory of the embedding system \cite{Coogan:2016, Invariance4MM}; identifying forward invariant and attractive sets by identifying equilibria in the embedding space \cite{Invariance4MM}; concluding global asymptotic stability by proving the nonexistence of equilibria of the embedding system except in a certain lower dimensional subspace \cite{smith2008global}.  See also \cite{smith2008monotone, monotonicity} for fundamental results on monotone dynamical systems.

A primary challenge in applying the theory of mixed-monotone systems is in identifying an appropriate decomposition function. There exists certain special cases for which a decomposition function can be readily identified, \emph{e.g.}, when each off-diagonal entry of the systems Jacobian matrix is uniformly upper or lower bounded \cite{7799445, SuffMM, TIRA, LTLandMM}, however, identifying decomposition functions generally relies on domain knowledge of the underlying physical system.

The question of existence of decomposition functions was recently explored in \cite{ozaything} in the discrete-time setting. In discrete-time, a decomposition function for an update map $F$ leads to an embedding system that \change{over-approximates} the image of $F$ when evaluated on a hyperrectangular set.  
\change{
It is observed in \cite{ozaything} that all discrete-time systems are mixed-monotone with a decomposition function that tightly approximates one-step reachable sets; this construction, however, fails to provide tight approximations for longer time horizons.}
While this result is constructive in that it provides an explicit decomposition function construction applicable to all discrete-time systems, evaluating the decomposition function at any point in the embedding space requires computing a reachable set itself.  
Nonetheless, knowing that a decomposition function does exist means that a search directed by, \emph{e.g.}, domain expertise, is not generally unreasonable.

In this paper, we study an analogous question regarding existence of decomposition functions in the continuous-time setting, and we additionally consider systems with disturbances. Our main result is to show that any continuous-time system possessing a vector field that is Lipschitz continuous in state and disturbance admits a Lipschitz continuous decomposition function.  
\change{
Moreover, we provide a construction for the decomposition function that yields the tightest possible approximations when used with the standard tools for mixed-monotone systems.
}
Thus, our results complement those from \cite{ozaything} by answering similar questions in the continuous-time setting, however, we emphasize that the results and tools here are different as compared to the discrete-time setting of \cite{ozaything}. In particular, unlike decomposition functions for continuous-time systems, decomposition functions for discrete-time systems do not need to be Lipschitz continuous, or even continuous. 
Moreover, we allow for disturbance inputs and define a different notion of tightness to accommodate the fact that it is generally not possible to obtain tight hyperrectangular reachable set approximations in continuous-time over any horizon. As in \cite{ozaything}, our construction is defined as an optimization problem, and thus not practically useful for applications other than system verification via simulation \cite{maxverifacation, dutreix2019specificationguided}.  However, we demonstrate through examples how tight decomposition functions can be \change{calculated} in closed form in certain instances.

As a second contribution, we show how under-approximations of reachable sets can be efficiently computed from a decomposition function for the backward-time dynamics.  Mixed-monotonicity in the backward-time setting was first considered in \cite{Invariance4MM} where it is shown how finite-time backward reachable sets can be approximated using an analogous technique to that of the forward-time case.  Here, we extend these results and specifically show that (a) a backward-time decomposition function can be used to compute under-approximations of forward reachable sets, (b) in certain instances, a tight backward-time decomposition function can be efficiently derived from a tight forward-time decomposition function, and (c) a tight backward-time decomposition function  provides the tightest, in a certain sense, under-approximations of forward reachable sets.

\section{Notation}\label{sec:notation}
Let $\mathbb{R}_{\geq 0}$ and $\mathbb{R}_{\leq 0}$ denote the nonnegative and nonpositive real numbers respectively. Let $\eR := \R\cup \{-\infty,\infty\}$ denote the extended real numbers.
Let $x_i$ for $i \in \{1,\cdots, n\}$ denote the $i^{\text{th}}$ entry of $x\in \R^n $.

Let $(x,\, y)$ denote the vector concatenation of $x,\, y \in \R^n$, \emph{i.e}. $(x,\,y) := [x^T \, y^T]^T \in \R^{2n}$, and let $\preceq$ denote the componentwise vector order, \emph{i.e.} $x\preceq y$ if and only if $x_i \leq y_i$ for all $i$.
Given $x, y\in\R^n$ with $x\preceq y$, 
\begin{equation*}
[x,\,y]:=\left\{z\in \R^n \,\mid\, x \preceq z \text{ and } z \preceq y\right \}
\end{equation*}
denotes the hyperrectangle defined by the endpoints $x$ and $y$.
We also allow $x_i,\,y_i\in\eR$ so that $[x,\,y]$ defines an \emph{extended hyperrectangle}, that is, a hyperrectangle with possibly infinite extent in some coordinates, \change{where componentwise order is extended to $\eR$ in the conventional way, \emph{i.e.}, $-\infty\leq x_i\leq \infty$ for all $x_i\in\eR$}.  
Given $a=(x,\,y) \in \R^{2n}$ with $x \preceq y$, we denote by $\rect{a}$ the hyperrectangle formed by the first and last $n$ components of $a$, \emph{i.e.},  $\rect{a}:=[x,\,y]$.

Let $\preceq_{\rm SE}$ denote the \emph{southeast order} on $\eR{}^{2n}$ defined by
\begin{equation*}
(x,\, x') \preceq_{\rm SE} (y,\, y')
 \:\: \Leftrightarrow  \:\:  x \preceq y\text{ and } y' \preceq x'
\end{equation*}
where $x,\, y,\, x',\, y' \in \eR{}^n$.  In the case that $x \preceq x'$ and $y \preceq y'$, note that
\begin{equation}
\label{eq:order_to_box}
(x,\, x') \preceq_{\rm SE} (y,\, y')
 \:\: \Leftrightarrow  \:\:
[\,y,\, y'\,] \subseteq [\,x,\, x'\,].
\end{equation}

\section{Preliminaries}\label{sec:preliminaries}
We consider the system 
\begin{equation}\label{eq1}
    \dot{x} = F(x,\, w)
\end{equation}
with state $x \in \X \subset \R^n$ and 
\change{time-varying disturbance input $w(t)\in \W \subset \R^m$.
We assume that the vector field $F:\X\times \W \rightarrow \R^n$ is locally Lipschitz continuous and that the disturbance signal $\mathbf{w}:\R \rightarrow \W$ is piecewise continuous, so that solutions to \eqref{eq1} are unique.
We also} assume that $\X$ is an extended hyperrectangle with nonempty interior and $\W$ is a hyperrectangle defined by $\W := [\underline{w},\, \overline{w}]$ for some $\underline{w},\, \overline{w} \in \R^m$ with $\underline{w} \preceq \overline{w}$.  Let $\Phi(T;\, x,\, \mathbf{w}) \in \X$ denote the state of \eqref{eq1}, reached at time $T$ when starting at the initial state $x \in \X$ at time $0$ and when subjected to the disturbance input $\mathbf{w}: [0,\, T]\rightarrow \W$. 
\change{
We allow for finite-time escape so that $\Phi(T;\, x,\, \mathbf{w})$ need not exist for all $T$, however, $\Phi(T;\, x,\, \mathbf{w})$ is understood to exist only when $\Phi(t;\, x,\, \mathbf{w}) \in \X$ for all $0\leq t \leq T$, and statements involving $\Phi(T;\, x,\, \mathbf{w})$ are understood to apply only when $\Phi(T;\, x,\, \mathbf{w})$ exists.
}

In this work, we are specifically interested in mixed-monotone systems.  Define by
\begin{equation}
    \begin{split}
        \T_{\X} &:= \{(x,\, \widehat{x}) \in \X \times \X  \:\vert\: x\preceq \widehat{x}\}, \\
        \T_{\W} &:= \{(w,\, \widehat{w}) \in \W \times \W  \:\vert\: w\preceq \widehat{w}\}, \\
    \end{split}
\end{equation}
the sets of ordered points in $\X$ and $\W$, respectively.  Additionally, define
\begin{equation}
\begin{split}
    \T &:= \{(x,\,w,\, \widehat{x},\, \widehat{w}) \in \X \times \W \times \X\times \W  \:\vert\: \\
        &\hspace{1.75cm} (x,\, \widehat{x}) \in \T_{\X} \text{ and } (w,\, \widehat{w}) \in \T_{\W}, \text{ or }  \\
        &\hspace{1.75cm} (\widehat{x},\, x) \in \T_{\X} \text{ and } (\widehat{w},\, w) \in \T_{\W} \}.
\end{split}
\end{equation}

\begin{definition}\label{def1}
Given a locally Lipschitz continuous function $d : \T \rightarrow \R^n$, the system \eqref{eq1} is \textit{mixed-monotone with respect to $d$}  if
\begin{enumerate}
    \item For all $x \in \X$ and all $w \in \W$ we have $d(x,\,w,\, x,\, w) = F(x,\, w)$.
    \item For all $i,\, j \in \{1,\, \cdots,\, n\}$, with $i \neq j$, we have $\frac{\partial d_i}{\partial x_j}(x,\,w,\, \widehat{x},\, \widehat{w}) \geq 0$ for all $(x,\,w,\, \widehat{x},\, \widehat{w}) \in \T$ such that $\frac{\partial d}{\partial x}$ exists.
    \item For all $i,\, j \in \{1,\, \cdots,\, n\}$, we have $\frac{\partial d_i}{\partial  \widehat{x}_j}(x,\,w,\,\widehat{x},\, \widehat{w}) \leq 0$ for all $(x,\,w,\, \widehat{x},\, \widehat{w}) \in \T$ such that $\frac{\partial d}{\partial \widehat{x}}$ exists.
    \item For all $i\in \{1,\, \cdots,\, n\}$ and all $k\in \{1,\, \cdots,\, m\}$, we have $\frac{\partial d_i}{\partial  w_j}(x,\,w,\,\widehat{x},\, \widehat{w}) \geq 0 \geq \frac{\partial d_i}{\partial  \widehat{w}_j}(x,\,w,\,\widehat{x},\, \widehat{w})$ for all $(x,\,w,\, \widehat{x},\, \widehat{w}) \in \T$ such that $\frac{\partial d}{\partial w}$ and $\frac{\partial d}{\partial \widehat{w}}$ exist.
    \qedhere
\end{enumerate}
\end{definition}

\change{
If \eqref{eq1} is mixed-monotone with respect to $d$, $d$ is said to be a \emph{decomposition function} for \eqref{eq1}, and derivatives of $d$ exist almost everywhere because $d$ is assumed Lipschitz.
Definition \ref{def1} is the standard definition of mixed-monotonicity for continuous-time dynamical systems and appears in, \emph{e.g.} \cite{Coogan:2016}, however, we make certain important generalisations here:
Previous works, \emph{e.g.}, \cite{Coogan:2016, SuffMM, TIRA}, consider decomposition functions with domain $\X\times\W\times\X\times\W$, however, we observe that the standard analysis tools for mixed-monotone systems---including those provided later---only require that $d$ be defined on $\T$. Thus, in Definition \ref{def1} we restrict the domain of $d$ to $\T$ without compromising the usefulness of the mixed-monotonicity property.
Additionally, we note that \cite{SuffMM} unnecessarily strengthens the conditions for continuous-time mixed-monotonicity and requires that Condition 2 from Definition 2 hold even in the case where $i = j$, and \cite{TIRA, LTLandMM} define mixed-monotone systems to be systems whose state and disturbance Jacobian matrices are uniformly bounded over the system domain. These are special cases of the more general conditions presented here.
}

\change{
\begin{remark}
\label{rem:kamke}
The above definition of mixed-monotonicity is in terms of the derivatives of the decomposition function $d$. By integrating, it can be shown that conditions 2--4 of Definition \ref{def1} are equivalent to the following two conditions:

\begin{enumerate}
    \item[C1)] For all $i\in\{1,\cdots,n\}$, 
      \begin{equation}
        \label{eq:8}
        d_i(x,\,w,\,\widehat{x},\,\widehat{w})\leq d_i(y,\,v,\,\widehat{x},\,\widehat{w})
      \end{equation}
      for all $(x,\, w,\, \widehat{x},\, \widehat{w}) \in \T$ and all $(y,\, v) \in \X \times \W$ such that  $(y,\,v,\,\widehat{x},\,\widehat{w}) \in \T$, $x\preceq y$, $x_i=y_i$ and $w \preceq v$.
    \item[C2)] For all $i\in\{1,\cdots,n\}$,
      \begin{equation}
        \label{eq:13}
        d_i(x,\,w,\,\widehat{y},\,\widehat{v}) \leq d_i(x,\,w,\,\widehat{x},\,\widehat{w})
      \end{equation}
      for all $(x,\, w,\, \widehat{x},\, \widehat{w}) \in \T$ and all $(\widehat{y},\, \widehat{v}) \in \X \times \W$ such that $(x,\,w,\,\widehat{y},\,\widehat{v}) \in \T$, $\widehat{x}\preceq \widehat{y}$ and  $\widehat{w}\preceq \widehat{v}$.
\end{enumerate}
These are the so-called \emph{Kamke} conditions for monotonicity, modified for the mixed-monotone setting \cite[Section 3]{smith2008monotone}.
\end{remark}
}

Construct 
\begin{equation}
\label{embedding}
    \begin{bmatrix}
    \dot{x} \\ \dot{\widehat{x}}
    \end{bmatrix}
    = e(x,\, \widehat{x}) = 
    \begin{bmatrix}
    d(x,\,\underline{w},\, \widehat{x},\,\overline{w}) \\ d(\widehat{x},\,\overline{w},\, x,\, \underline{w})
    \end{bmatrix}
\end{equation}
with state $(x,\, \widehat{x}) \in \T_{\X}$.
We refer to \eqref{embedding} as the \textit{embedding system relative to $d$} and $e$ the \textit{embedding function relative to $d.$}  Let $\Phi^e(T;\, (x,\, \widehat{x}))$ be the state transition function for \eqref{embedding}, that is, $\Phi^e(T;\, (x,\, \widehat{x}))$ denotes the state of \eqref{embedding} at time $T$ when initialised at state $(x,\, \widehat{x})\in \T_{\X}$. 
\change{
Trajectories of the embedding system may leave $\X\times \X$, and this is true even when 
$\Phi(t;\, x_0,\, \mathbf{w}) \in \X$ for all $t\geq 0$, all $x_0\in X$, and all $\mathbf{w}: \R \rightarrow \W$.
However, trajectories of \eqref{embedding} will not evolve from $\T_{\X}$ to $(\X\times \X)\backslash \T_{\X}$ \cite[Lemma 1]{Invariance4MM}.
}
\change{For this reason}, $\Phi^e(T;\, (x,\, \widehat{x}))$ is understood to exist only when $\Phi^e(t;\, (x,\, \widehat{x}))\in \T_{\X}$ for all $0\leq t \leq T$, and statements involving $\Phi^e(T;\, (x,\, \widehat{x}))$ are understood to apply only when $\Phi^e(T;\, (x,\, \widehat{x}))$ exists.  Importantly, \eqref{embedding} is \emph{monotone with respect to the southeast order}; that is, for all $(x,\, \widehat{x}),\, (y,\, \widehat{y})\in \T_{\X}$ and all $T \geq 0$ we have
\begin{equation*}
    (x,\, \widehat{x}) \preceq_{\rm SE} (y,\, \widehat{y}) \Rightarrow \Phi^e(T;\, (x,\, \widehat{x})) \preceq_{\rm SE} \Phi^e(T;\, (y,\, \widehat{y})).
\end{equation*}

We next recall how reachable sets for \eqref{eq1} are over-approximated by trajectories of \eqref{embedding}.  To that end, denote by
\begin{multline}\label{fwd_reach}
    R^{F}(T;\, \X_0) = \{ \Phi(T;\, x,\, \mathbf{w}) \in \X \:\vert\: x \in \X_0,\, \\
    \text{for some }\mathbf{w}:[0,\, T]\rightarrow \W\}
\end{multline}
the forward reachable set of \eqref{eq1} over the time horizon $T$ from the set of initial conditions $\X_0 \subset \X$.
The following fundamental result connects reachable sets to the dynamics of the embedding system \cite{Coogan:2016, TIRA}.

\begin{proposition}\label{prop:p1} 
 Let \eqref{eq1} be mixed-monotone with respect to $d$, and let $\X_0 = [\underline{x},\, \overline{x}]$ for some $\underline{x}\preceq \overline{x}$.
Then $R^{F}(T;\, \X_0) \subseteq \rect{\Phi^{e}( T;\, (\underline{x},\, \overline{x}))}$.
\end{proposition}

\change{The system 
\begin{equation}\label{backward}
    \dot{x} = -F(x,\, w)
\end{equation}
with $x \in \X$ and $w \in \W$ encodes the backward-time dynamics of \eqref{eq1}, and \eqref{eq1} and \eqref{backward} are related in the following way: if $x_1 = \Phi(T;\, x_0,\, \mathbf{w})$ for $\mathbf{w}:[0,\, T] \rightarrow \W$, then $x_0 = \Phi'(T;\, x_1,\, \mathbf{w}')$ where $\mathbf{w}'(t) := \mathbf{w}(T - t)$ and where $\Phi'$ denotes the state transition function of \eqref{backward}. 
Therefore, if \eqref{backward} is mixed-monotone, then finite-time backward reachable sets of \eqref{eq1} can be approximated using \cite[Proposition 2]{Invariance4MM} and this procedure is analogous to that of Proposition \ref{prop:p1}.}

\section{Tight Decomposition functions For Mixed-Monotone Systems}\label{sec:mainresult}
In this section, we show that all continuous-time dynamical systems with disturbances as in \eqref{eq1} are mixed-monotone, and we provide an explicit construction for the decomposition function that provides the tightest reachable set approximations via Proposition \ref{prop:p1}.

\begin{definition}[Tight Decomposition Function]\label{def:tightdecomp}
A decomposition function $\delta$ for \eqref{eq1} is \emph{tight} if for any other decomposition function $d$ for \eqref{eq1}, 
\begin{equation}\label{eq:tight}
    \begin{split}
        d(x,\,w,\,\widehat{x},\,\widehat{w}) &\preceq \delta(x,\,w,\,\widehat{x},\,\widehat{w}) \\
        \delta(\widehat{x},\,\widehat{w},\,x,\,w) &\preceq d(\widehat{x},\,\widehat{w},\,x,\,w)
    \end{split}
\end{equation}
for all $(x,\, \widehat{x}) \in \T_{\X}$ and all $({w},\, \widehat{w}) \in \T_{\W}$.
\end{definition}

\change{
As we show next, a tight decomposition function, when used with Proposition \ref{prop:p1}, will provide a tighter over-approximation of reachable sets than of any other decomposition function.
}

\begin{proposition}\label{prop:p2}
If $d$ is a decomposition function for \eqref{eq1} and $\delta$ is a tight decomposition function for \eqref{eq1}, then for all $t \geq 0$
\begin{equation}
\label{prop:res}
    \rect{\Phi^{\varepsilon}(t;\, (\underline{x},\, \overline{x}))} \subseteq \rect{\Phi^e(t;\, (\underline{x},\, \overline{x}))}
\end{equation}
for all $(\underline{x},\, \overline{x}) \in \T_{\X}$ where $\Phi^{\varepsilon}$ and $\Phi^e$ denote the state transition functions of the embedding system \eqref{embedding} constructed from $\delta$ and $d$, respectively.
\end{proposition}
\begin{proof}
Let $\delta$ and $d$ be such that \eqref{eq:tight} holds.  Let $\varepsilon$ and $e$ denote the embedding functions relative to $\delta$ and $d$, respectively, and let
$\Phi^{\varepsilon}$ and $\Phi^e$ denote the state transition functions of their respective embedding systems.  Choose $(\underline{x},\, \overline{x}) \in \T_{\X}$, and define 
\begin{equation}
\varphi^{e}(t) = \Phi^{e}(t;\, (\underline{x},\, \overline{x})), \:\text{ and }\:\varphi^{\varepsilon}(t) = \Phi^{\varepsilon}(t;\, (\underline{x},\, \overline{x})),
\end{equation}
where we write $\varphi^e =: (\underline{\varphi}^e,\,\overline{\varphi}^e)$ and $\varphi^{\varepsilon} =: (\underline{\varphi}^{\varepsilon},\,\overline{\varphi}^{\varepsilon})$. Then 
\begin{equation}
    \dot{\varphi}^e = e(\underline{\varphi}^{e},\,\overline{\varphi}^{e}), \:\text{ and }\: \dot{\varphi}^{\varepsilon} = \varepsilon(\underline{\varphi}^{\varepsilon},\, \overline{\varphi}^{\varepsilon}).
\end{equation}
We now show that $\varphi^e(0)\preceq_{\rm SE} \varphi^{\varepsilon}(0)$ implies $\varphi^e(t)\preceq_{\rm SE} \varphi^e(t)$ for all $t\geq 0$. Assume there exists a time $T \geq 0$ such that 
\begin{equation}\label{eq:cond}
    \varphi_i^e(T) = \varphi_i^{\varepsilon}(T) \:\text{ and }\: \varphi^e(T) \preceq_{\rm SE} \varphi^{\varepsilon}(T)
\end{equation}
for some $i\in\{1,\cdots,2n\}$. Consider first the case that $i \in \{1,\,\cdots,\, n\}$. Then
\begin{equation}\label{kamke}
    \begin{split}
        d_i( \underline{\varphi}^e(T),\, \underline{w},\, \overline{\varphi}^e(T),\, \overline{w})
        & \leq
        \delta_i( \underline{\varphi}^e(T),\, \underline{w},\, \overline{\varphi}^e(T),\, \overline{w}), \\
        & \leq
        \delta_i( \underline{\varphi}^{\varepsilon}(T),\, \underline{w},\, \overline{\varphi}^{\varepsilon}(T),\, \overline{w}),
    \end{split}
\end{equation}
where the first inequality comes from the fact that $\delta$ is a tight decomposition function for \eqref{eq1}, and where the second inequality comes from Conditions C1 and C2 in Remark \ref{rem:kamke}. 
Thus we now have $\dot{\varphi}_i^{e}(T) \leq \dot{\varphi}_i^{\varepsilon}(T).$
If instead \eqref{eq:cond} holds for some $i\in\{n+1,\cdots,2n\}$, by a symmetric argument, $\dot{\varphi}_i^{e}(T) \geq \dot{\varphi}_i^{\varepsilon}(T)$.  
Therefore, always $\varphi^e(t)\preceq_{\rm SE} \varphi^{\varepsilon}(t)$, which is equivalently to \eqref{prop:res} by \eqref{eq:order_to_box}. This completes the proof.
\end{proof}

In the following theorem, we show that all continuous-time systems with disturbances as in \eqref{eq1} are mixed-monotone and we present a construction for tight decomposition functions.

\begin{theorem}\label{thm1}
Any system of the form \eqref{eq1} is mixed-monotone with respect to $\delta:\T \rightarrow \R^n$ constructed elementwise according to
\begin{multline}\label{opt_decomp}
    \delta_i(x,\, w,\, \widehat{x},\, \widehat{w}) = \\
    \begin{cases}
    \min\limits_{
    \substack{
        y \in  [x,\, \widehat{x}]\\
        y_i = x_i \\
        z\in [w,\, \widehat{w}]
    }
    }
    F_i(y,\, z) & \text{if $x \preceq \widehat{x}$ and $w \preceq \widehat{w},$} \\
    \vspace{-.3cm}
    \\
    \max\limits_{
    \substack{ y \in  [\widehat{x},\, x]\\
               y_i = x_i \\
               z \in [\widehat{w},\, w]
               }
    } 
    F_i(y,\, z) & \text{if $\widehat{x} \preceq x$ and $\widehat{w} \preceq w$}.
    \end{cases}
\end{multline}
Moreover, $\delta$ is a tight decomposition function for \eqref{eq1}.
\end{theorem}
\begin{proof}
We begin by establishing that $\delta$ from \eqref{opt_decomp} is Lipschitz continuous; this is done by showing that $\delta_i$ is Lipschitz in $x$, and Lipschitz continuity holds with respect to other arguments by analogous reasoning.
Let $x^1, x^2,\widehat{x}\in\X$,  $w,\widehat{w}\in \W$, where we assume without loss of generality that $x^1\preceq \widehat{x}$, $x^2\preceq \widehat{x}$, and $w\preceq \widehat{w}$. Observe that for any $y^1\in [x^1,\widehat{x}]$ with $y^1_i=x^1_i$, there exists $y^2\in[x^2,\widehat{x}]$ with $y^2_i=x^2_i$ such that $\|y^1-y^2\|_1\leq \|x^1-x^2\|_{1}$, and vice-versa, where $\|\cdot\|_{1}$ denotes the usual one-norm on $\mathbb{R}^n$. In particular, for any minimizer $(y^1,z)$ that achieves the value of $\delta_i(x^1,w,\widehat{x},\widehat{w})$ in the definition \eqref{opt_decomp}, there exists a point $y^2$ so that $F_i(y^2,z)$ upper bounds $\delta_i(x^2,w,\widehat{x},\widehat{w})$ with $\|y^1-y^2\|_1\leq \|x^1-x^2\|_1$, and vice-versa.  It follows then that $\|\delta_i(x^1,w,\widehat{x},\widehat{w})-\delta_i(x^2,w,\widehat{x},\widehat{w})\|_1 \leq L\|x^1-x^2\|_{1}$ where $L$ is a  Lipschitz constant for $F$ applicable on a neighborhood of $[x^1,\widehat{x}]\cup [x^2,\widehat{x}]$.  Thus $\delta_i$ is Lipschitz in $x$, and therefore $\delta$ is Lipschitz in $x,\, \widehat{x},\, w,\, \widehat{w}$.

We next show that $\delta$ is a decomposition function for \eqref{eq1}. 
Trivially, $\delta_{i}(x, w, x, w)=F_i(x,w)$ for all $i$ and all $x\in \X$, $w\in\W$. 
We show that $\delta$ satisfies Conditions 2--4 from Definition \ref{def1} by showing that $\delta_i$ satisfies the Kamke conditions in Remark \ref{rem:kamke}. 
\change{Choose $(x,\, w,\, \widehat{x},\, \widehat{w}) \in \T$ and $(y,\, v) \in \X\times\W$ such that $(y,\, v,\, \widehat{x},\, \widehat{w}) \in \T$, $x\preceq y$, $x_i=y_i$, and $w\preceq v$}. Then $\delta_i(x, w, \widehat{x}, \widehat{w})\leq \delta_i(y, v, \widehat{x}, \widehat{w})$ follows from the $\mbox{min}/\mbox{max}$ construction of \eqref{opt_decomp}.  This proves C1, and C2 is proven analogously.
\textit{Thus}, \eqref{eq1} is mixed-monotone with respect to $\delta$.

Lastly, we show that $\delta$ is a tight decomposition function for \eqref{eq1}. Let $d:\T \rightarrow \R^n$ be another decomposition function for \eqref{eq1} and choose $(\underline{x},\, \overline{x})\in \T_{\X}$ and $(\underline{w},\, \overline{w})\in \T_{\W}$.  Additionally, choose $x \in [\underline{x},\, \overline{x}]$ and $w\in [\underline{w},\, \overline{w}]$.  Then $(\underline{x},\,\overline{x}) \preceq_{\rm SE} (x,\,x)$ and $(\underline{w},\,\overline{w}) \preceq_{\rm SE} (w,\,w)$, and therefore
\begin{equation}\label{corbin}
    \begin{bmatrix}
    d(\underline{x}, \underline{w}, \overline{x}, \overline{w}) \\
    d(\overline{x}, \overline{w}, \underline{x}, \underline{w})
    \end{bmatrix}
    \preceq_{\rm SE}
    \begin{bmatrix}
    d(x, w, x, w) \\
    d(x, w, x,w)
    \end{bmatrix}
    =
    \begin{bmatrix}
    F(x, w) \\
    F(x, w)
    \end{bmatrix}.
\end{equation}
Since \eqref{corbin} holds for all  $x \in [\underline{x},\, \overline{x}]$ and all $w\in [\underline{w},\, \overline{w}]$ we now have
\begin{equation*}
    \begin{bmatrix}
    d(\underline{x}, \underline{w}, \overline{x}, \overline{w}) \\
    d(\overline{x}, \overline{w}, \underline{x}, \underline{w})
    \end{bmatrix}
    \preceq_{\rm SE}
    \begin{bmatrix}
    \min\limits_{
    \substack{
        y \in  [\underline{x},\, \overline{x}], y_i = \underline{x}_i, z\in [\underline{w},\, \overline{w}] \\{}
    }
    }
    F(y,\,z) \\
    \max\limits_{
    \substack{
        y \in  [\underline{x},\, \overline{x}], y_i = \underline{x}_i, z\in [\underline{w},\, \overline{w}]
    }
    }
    F(y,\,z) \\
    \end{bmatrix}, 
\end{equation*}
and thus 
\begin{equation*}
    \begin{bmatrix}
    d(\underline{x}, \underline{w}, \overline{x}, \overline{w}) \\
    d(\overline{x}, \overline{w}, \underline{x}, \underline{w})
    \end{bmatrix}
    \preceq_{\rm SE}
    \begin{bmatrix}
    \delta(\underline{x}, \underline{w}, \overline{x}, \overline{w}) \\
    \delta(\overline{x}, \overline{w}, \underline{x}, \underline{w})
    \end{bmatrix}.
\end{equation*}
Therefore $\delta$ is a tight decomposition function for \eqref{eq1} as $(\underline{x},\, \overline{x})\in \T_{\X}$ and $(\underline{w},\, \overline{w})\in \T_{\W}$ were selected arbitrarily. This completes the proof.
\end{proof}

We next demonstrate the applicability of Theorem \ref{thm1} through an example.

\change{
\begin{example}\label{ex1}
The system 
\begin{equation}\label{examsys}
    \begin{bmatrix}
    \dot{x}_1 \\
    \dot{x}_2
    \end{bmatrix}
     = F(x) = 
    \begin{bmatrix}
    \:|x_1 - x_2| \\
    -x_1
    \end{bmatrix}
\end{equation}
with $\X = \R^2$ is mixed-monotone with respect to
\begin{equation}\label{examdec}
\begin{split}
    \delta_1(x,\, \widehat{x}) &= \begin{cases}
    0 & 
    \text{if $x_2 \leq x_1 \leq \widehat{x}_2$}, \\
    x_2 - x_1 & 
    \text{if $2x_1\leq \mbox{min}\{ 2x_2,\, x_2+\widehat{x}_2$\} }, \\
    x_1 - \widehat{x}_2 &
    \text{if $2 x_1 \geq \mbox{min}\{ 2\widehat{x}_2,\, x_2+\widehat{x}_2$\},}
    \end{cases} \\
    \delta_2(x,\, \widehat{x}) &= -\widehat{x}_1,
\end{split}
\end{equation}
where $\delta$ is a tight decomposition function and solves \eqref{opt_decomp}.

To further contrast our results to their discrete-time analog in \cite{ozaything}, an alternative decomposition function is obtained by applying the construction presented in  \cite[Theorem 1]{ozaything}, which is a tight decomposition for the discrete-time system $x^+=F(x)$, but is generally not a tight decomposition function for the continuous-time system $\dot{x}=F(x)$, as demonstrated in Figure \ref{fig2}.
\end{example}
}

\begin{figure}
%
%
\definecolor{mycolor1}{rgb}{1.00000,0.00000,1.00000}%
\begin{tikzpicture}

\begin{axis}[%
width=5.238cm,
height=3.5cm,
at={(0cm,0cm)},
scale only axis,
xmin=-1.5,
xmax=4.5,
xtick={-1.5,    0,  1.5,    3,  4.5},
xlabel style={font=\color{white!15!black}},
xlabel={$x_1$},
ymin=-2.5,
ymax=2.5,
ytick={-2.5,    0,  2.5},
ylabel style={font=\color{white!15!black}},
ylabel={$x_2$},
axis background/.style={fill=white},
axis x line*=bottom,
axis y line*=left,
xmajorgrids,
ymajorgrids,
axis on top
]

\addplot[area legend, draw=black, fill=mycolor1, fill opacity=0.2, forget plot]
table[row sep=crcr] {%
x	y\\
-1	-2.30657041083348\\
4.32953773011248	-2.30657041083348\\
4.32953773011248	2.01\\
-1	2.01\\
}--cycle;

\addplot[area legend, draw=black, fill=white, forget plot]
table[row sep=crcr] {%
x	y\\
-0.731019083635831	-2.02327436916925\\
3.80837226572762	-2.02327436916925\\
3.80837226572762	1.78788243552011\\
-0.731019083635831	1.78788243552011\\
}--cycle;

\addplot[area legend, draw=black, fill=blue, fill opacity=0.2, forget plot]
table[row sep=crcr] {%
x	y\\
-0.731019083635831	-2.02327436916925\\
3.80837226572762	-2.02327436916925\\
3.80837226572762	1.78788243552011\\
-0.731019083635831	1.78788243552011\\
}--cycle;

\addplot[area legend, draw=black, fill=white, forget plot]
table[row sep=crcr] {%
x	y\\
-0.116917991093013	0.538060338394386\\
-0.104610834135854	0.481422408037082\\
-0.0923036771786944	0.424784477679779\\
-0.0799965202215351	0.368146547322475\\
-0.0676893632643759	0.311508616965171\\
-0.0553822063072166	0.254870686607867\\
-0.0430750493500574	0.198232756250563\\
-0.0307678923928981	0.14159482589326\\
-0.0184607354357389	0.0849568955359557\\
-0.00615357847857962	0.0283189651786519\\
0.200440645564611	-0.106488124693118\\
0.601321936693834	-0.319464374079356\\
1.00220322782306	-0.532440623465593\\
1.40308451895228	-0.745416872851829\\
1.8039658100815	-0.958393122238067\\
2.20484710121073	-1.1713693716243\\
2.60572839233995	-1.38434562101054\\
3.00660968346917	-1.59732187039678\\
3.4074909745984	-1.81029811978301\\
3.80837226572762	-2.02327436916925\\
3.7018841410345	-1.92932184829776\\
3.59539601634138	-1.83536932742626\\
3.48890789164826	-1.74141680655477\\
3.38241976695514	-1.64746428568328\\
3.27593164226203	-1.55351176481179\\
3.16944351756891	-1.45955924394029\\
3.06295539287579	-1.3656067230688\\
2.95646726818267	-1.27165420219731\\
2.84997914348955	-1.17770168132582\\
2.74349101879643	-1.08374916045432\\
2.63700289410331	-0.989796639582829\\
2.53051476941019	-0.895844118711337\\
2.42402664471708	-0.801891597839845\\
2.31753852002396	-0.707939076968351\\
2.21105039533084	-0.613986556096859\\
2.10456227063772	-0.520034035225365\\
1.9980741459446	-0.426081514353873\\
1.89158602125148	-0.33212899348238\\
1.78509789655836	-0.238176472610886\\
1.42870332258613	-0.0484661437507417\\
1.14903307174389	0.102446929656295\\
0.937816552714142	0.220943829009225\\
0.785095310805562	0.314305650156099\\
0.68016337982962	0.389937266223042\\
0.612827612173928	0.454474300891685\\
0.574167624405822	0.513237603786257\\
0.556521073830125	0.570021433951443\\
0.544213916872966	0.626659364308747\\
0.531906759915807	0.683297294666051\\
0.519599602958648	0.739935225023355\\
0.507292446001488	0.796573155380658\\
0.494985289044329	0.853211085737963\\
0.48267813208717	0.909849016095266\\
0.47037097513001	0.96648694645257\\
0.458063818172851	1.02312487680987\\
0.445756661215692	1.07976280716718\\
0.433449504258533	1.13640073752448\\
0.421142347301374	1.19303866788179\\
0.392823382122721	1.15856612422455\\
0.36450441694407	1.12409358056732\\
0.336185451765418	1.08962103691009\\
0.307866486586766	1.05514849325286\\
0.279547521408114	1.02067594959563\\
0.251228556229462	0.986203405938396\\
0.22290959105081	0.951730862281165\\
0.194590625872158	0.917258318623933\\
0.166271660693506	0.882785774966702\\
0.137952695514854	0.84831323130947\\
0.109633730336202	0.813840687652239\\
0.0813147651575505	0.779368143995007\\
0.0529957999788986	0.744895600337774\\
0.0246768348002468	0.710423056680543\\
-0.00364213037840535	0.675950513023312\\
-0.031961095557057	0.641477969366081\\
-0.0602800607357091	0.607005425708849\\
-0.088599025914361	0.572532882051618\\
-0.116917991093013	0.538060338394386\\
}--cycle;

\addplot[area legend, draw=black, fill=red, forget plot]
table[row sep=crcr] {%
x	y\\
-1	0\\
1	0\\
1	1\\
-1	1\\
}--cycle;

\addplot[area legend, draw=black, fill=green, fill opacity=0.8, forget plot]
table[row sep=crcr] {%
x	y\\
-0.116917991093013	0.538060338394386\\
-0.104610834135854	0.481422408037082\\
-0.0923036771786944	0.424784477679779\\
-0.0799965202215351	0.368146547322475\\
-0.0676893632643759	0.311508616965171\\
-0.0553822063072166	0.254870686607867\\
-0.0430750493500574	0.198232756250563\\
-0.0307678923928981	0.14159482589326\\
-0.0184607354357389	0.0849568955359557\\
-0.00615357847857962	0.0283189651786519\\
0.200440645564611	-0.106488124693118\\
0.601321936693834	-0.319464374079356\\
1.00220322782306	-0.532440623465593\\
1.40308451895228	-0.745416872851829\\
1.8039658100815	-0.958393122238067\\
2.20484710121073	-1.1713693716243\\
2.60572839233995	-1.38434562101054\\
3.00660968346917	-1.59732187039678\\
3.4074909745984	-1.81029811978301\\
3.80837226572762	-2.02327436916925\\
3.7018841410345	-1.92932184829776\\
3.59539601634138	-1.83536932742626\\
3.48890789164826	-1.74141680655477\\
3.38241976695514	-1.64746428568328\\
3.27593164226203	-1.55351176481179\\
3.16944351756891	-1.45955924394029\\
3.06295539287579	-1.3656067230688\\
2.95646726818267	-1.27165420219731\\
2.84997914348955	-1.17770168132582\\
2.74349101879643	-1.08374916045432\\
2.63700289410331	-0.989796639582829\\
2.53051476941019	-0.895844118711337\\
2.42402664471708	-0.801891597839845\\
2.31753852002396	-0.707939076968351\\
2.21105039533084	-0.613986556096859\\
2.10456227063772	-0.520034035225365\\
1.9980741459446	-0.426081514353873\\
1.89158602125148	-0.33212899348238\\
1.78509789655836	-0.238176472610886\\
1.42870332258613	-0.0484661437507417\\
1.14903307174389	0.102446929656295\\
0.937816552714142	0.220943829009225\\
0.785095310805562	0.314305650156099\\
0.68016337982962	0.389937266223042\\
0.612827612173928	0.454474300891685\\
0.574167624405822	0.513237603786257\\
0.556521073830125	0.570021433951443\\
0.544213916872966	0.626659364308747\\
0.531906759915807	0.683297294666051\\
0.519599602958648	0.739935225023355\\
0.507292446001488	0.796573155380658\\
0.494985289044329	0.853211085737963\\
0.48267813208717	0.909849016095266\\
0.47037097513001	0.96648694645257\\
0.458063818172851	1.02312487680987\\
0.445756661215692	1.07976280716718\\
0.433449504258533	1.13640073752448\\
0.421142347301374	1.19303866788179\\
0.392823382122721	1.15856612422455\\
0.36450441694407	1.12409358056732\\
0.336185451765418	1.08962103691009\\
0.307866486586766	1.05514849325286\\
0.279547521408114	1.02067594959563\\
0.251228556229462	0.986203405938396\\
0.22290959105081	0.951730862281165\\
0.194590625872158	0.917258318623933\\
0.166271660693506	0.882785774966702\\
0.137952695514854	0.84831323130947\\
0.109633730336202	0.813840687652239\\
0.0813147651575505	0.779368143995007\\
0.0529957999788986	0.744895600337774\\
0.0246768348002468	0.710423056680543\\
-0.00364213037840535	0.675950513023312\\
-0.031961095557057	0.641477969366081\\
-0.0602800607357091	0.607005425708849\\
-0.088599025914361	0.572532882051618\\
-0.116917991093013	0.538060338394386\\
}--cycle;
\end{axis}
\end{tikzpicture}%
    \caption{
    \change{
    Approximating forward reachable sets for \eqref{examsys} from the set of initial conditions $\X_0 = [-1,\, 1]\times [0,\, 1]$, shown in red.  $R^{F}(1;\, \X_0)$ is computed via exhaustive simulation and is shown in green.
    Hyperrectangular over-approximations of $R^{F}(1/2;\, \X_0)$ are computed from \eqref{examdec} and \cite[Theorem 1]{ozaything} and are shown in blue and pink, respectively.
    }
    }
    \label{fig2}
\end{figure}
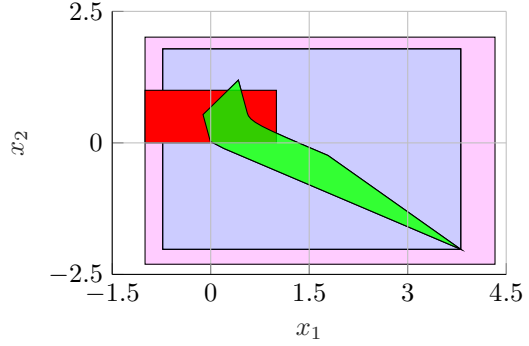

\section{Under-Approximating Reachable Sets via Mixed-Monotonicity}
\change{
We next show how under-approximations of reachable sets are computed via the mixed-monotonicity property. As in \eqref{fwd_reach}, let $R^{F}(T;\, \X_0)$ denote the time-$T$ forward reachable set of \eqref{eq1} from the hyperrectangular set of initial conditions $\X_0$.
}

\begin{theorem}\label{thm2}
Let \eqref{backward} be mixed-monotone with respect to $D$, and let $\X_0 = [\underline{x},\, \overline{x}]$ for some $\underline{x}\preceq \overline{x}$. Construct the system
\begin{equation}\label{under_embedding}
    \begin{bmatrix}
    \dot{x} \\ \dot{\widehat{x}}
    \end{bmatrix}
    =
    \Gamma(x,\, \widehat{x})
    =
    \begin{bmatrix}
    -D(x,\, \underline{w},\, \widehat{x},\, \overline{w}) \\ 
    -D(\widehat{x},\, \overline{w},\, x,\, \underline{w})
    \end{bmatrix}
\end{equation}
with state transition function $\Phi^{\Gamma}$.  If $\Phi^{\Gamma}(t;\, (\underline{x},\, \overline{x})) \in \T_{\X}$ for all $0 \leq t \leq T$ then $\rect{\Phi^{\Gamma}(T;\, (\underline{x},\, \overline{x}))} \subseteq R^{F}(T;\, \X_0).$
\end{theorem}
\begin{proof}
Let $E$ denote the embedding function relative to $D$ and let $\Phi^{E}$ denote the state transition function of its embedding system.  \change{Then for all $(x,\, \widehat{x}),\, (y,\, \widehat{y}) \in \T_{\X}$ and all $T \geq 0$ we have that $\Phi^{E}(T;\, (x,\, \widehat{x})) = (y,\, \widehat{y})$ if and only if $\Phi^{\Gamma}(T;\, (y,\, \widehat{y})) = (x,\, \widehat{x})$.}

We prove \change{Theorem \ref{thm2}} by showing that for all $y \in \rect{\Phi^{\Gamma}(T; (\underline{x},\, \overline{x}))}$ there exists an $x \in \X_0$ and a disturbance input $\mathbf{w}:[0,\, T] \rightarrow \W$ such that $y = \Phi(T;\, x,\, \mathbf{w})$.
Define $\varphi(t) := \Phi^{\Gamma}(t;\, (\underline{x},\, \overline{x}))$ where we let $\varphi(t) =: (\underline{\varphi}(t),\, \overline{\varphi}(t))$.  Choose $T \geq 0$ and $y \in \rect{\Phi^{\Gamma}(T; (\underline{x},\, \overline{x}))}=[\underline{\varphi}(T),\, \overline{\varphi}(T)]$ where we have $(\underline{\varphi}(T),\, \overline{\varphi}(T)) \preceq_{\rm SE} (y,\, y)$ by \eqref{eq:order_to_box}.  Then 
\begin{equation}
    \Phi^{E}(T;\, (\underline{\varphi}(T),\, \overline{\varphi}(T))) \preceq_{\rm SE} \Phi^{E}(T;\, (y,\, y))
\end{equation}
follows from the monotonicity of the embedding system relative to $E$. 
As a result of Proposition \ref{prop:p1} we have that for any $\mathbf{w}':[0,\, T] \rightarrow \W$,  $\Phi'(T;\, y,\, \mathbf{w}')\in \Phi^{E}(T;\, (\underline{\varphi}(T),\, \overline{\varphi}(T))) = [\underline{x},\, \overline{x}]$ where $\Phi'$ is taken to be the state transition function of \eqref{backward}. Take $x=\Phi'(T;\, y,\, \mathbf{w}')$ and define $\mathbf{w}(t) := \mathbf{w}'(T-t)$. Then $y = \Phi(T;\, x,\, \mathbf{w})$.  This completes the proof.
\end{proof}

\change{
While $\Gamma$ from \eqref{under_embedding} is constructed from the decomposition function of the backward-time dynamics \eqref{backward}, trajectories of \eqref{under_embedding} can evolve from $\T_{\X}$ to $(\X\times \X)\backslash \T_{\X}$, unlike trajectories of \eqref{embedding}.
For this reason, $\rect{\Phi^{\Gamma}(T;\, (\underline{x},\, \overline{x}))} \subseteq R^{F}(T;\, [\underline{x},\, \overline{x}])$ only if $\Phi^{\Gamma}(t;\, (\underline{x},\, \overline{x})) \in \T_{\X}$ for all $0 \leq t \leq T$.
}

In Theorem \ref{thm2} we show how the system \eqref{under_embedding}, which is constructed from a decomposition function for \eqref{backward}, is used to under-approximate forward reachable sets for the system \eqref{eq1}.  As a consequence of Theorem \ref{thm1} we have that \eqref{backward} is mixed-monotone, and we next provide a special case for which a tight decomposition function for the backward-time system \eqref{backward} can be computed from a tight decomposition function for the forward-time system \eqref{eq1}.

\begin{specialblock}\label{spec2}
If
\begin{enumerate}
\item $\delta$ is a tight decomposition function for \eqref{eq1}, and
\item $F_i$ does not depend on $x_i$ for all $i\in\{1,\,\cdots,\, n\}$,
\end{enumerate}
then $\Delta(x,\, w,\, \widehat{x},\, \widehat{w}) := - \delta(\widehat{x},\, \widehat{w},\, x,\, w)$ is a tight decomposition function for \eqref{backward}.
\end{specialblock}

To summarise the previous results, the tight decomposition function $\delta$ from \eqref{opt_decomp} allows one to compute tight over-approximations of forward reachable sets via Proposition \ref{prop:p1}.  If $F$ satisfies the hypothesis of Special Case \ref{spec2}, then $\Delta(x,\, w,\, \widehat{x},\, \widehat{w}) := -\delta(\widehat{x},\, \widehat{w},\, x,\, w)$ allows computing over-approximations of backward reachable sets for \eqref{eq1} via \cite[Proposition 2]{Invariance4MM}, and, by analogous reasoning to that of Proposition \ref{prop:p2}, it can be additionally shown that $\Delta$ provides the \emph{tightest} possible over-approximations of backward reachable sets.  Last, note that $\Delta$ provides large under-approximations of reachable sets when used with Theorem \ref{thm2}, and thus, for systems satisfying the hypothesis of Special Case \ref{spec2}, implementing the reachability tools detailed in this paper requires only requires one computation of \eqref{opt_decomp} for each state.

\change{
\section{Numerical Example}
}
\change{The system 
\begin{equation}\label{cs_sys}
    \begin{bmatrix}
    \dot{x}_1 \\ \dot{x}_2 \\ \dot{x_3}
    \end{bmatrix}
    = F(x,w) =
    \begin{bmatrix}
    w_1x_2^2 - x_2 + w_2 \\  
    x_3 + 2\\
    x_1 - x_2 - w_1^3
    \end{bmatrix}
\end{equation}
with $\X = \R^3$ and $\W \subset \R^2$ is mixed-monotone with respect to 
\begin{multline}\label{max7}
    \delta_1(x,\, w,\, \widehat{x},\, \widehat{w}) = \\
    \begin{cases}
    \frac{-1\:\:}{4 w_1} + w_2 
    & \text{if $w_1 x_2 \leq \tfrac{1}{2} \leq w_1\widehat{x}_2$,} \\
    w_1 x_2^2 - x + w_2
    & \text{if $\tfrac{1}{2} \leq w_1 x_2 \leq w_1\widehat{x}_2$,} \\
    & \text{or $1 \leq w_1 (x_2 + \widehat{x}_2)$,} \\
    w_1 \widehat{x}_2^2 - \widehat{x} + w_2
    & \text{if $w_1 x_2 \leq w_1 \widehat{x}_2 \leq \tfrac{1}{2}$,} \\
    & \text{or $w_1 (x_2 + \widehat{x}_2) \leq 1$,} \\
    \end{cases}
\end{multline}
\begin{equation}\label{max4}
    \begin{split}
        \delta_2(x,\,w,\, \widehat{x},\, \widehat{w}) 
        &= x_3 + 2,\\
        \delta_3(x,\,w,\, \widehat{x},\, \widehat{w}) 
        &= x_1 - \widehat{x}_2 - \widehat{w}_1^3.
    \end{split}
\end{equation}
where $\delta$ is a tight decomposition function and solves \eqref{opt_decomp}.
The second and third components of $\delta$, given in \eqref{max4}, are straightforwardly derived from \eqref{opt_decomp}, and we justify the construction of $\delta_1$ as follows:  The minimum, or maximum, of a scalar-valued function will either occur on the boundary of the function's domain or at a critical point in the interior of the domain. Note that the optimization problem \eqref{opt_decomp} is evaluated over a hyperrectangle, and the boundary of this domain is also comprised of hyperrectangles.  Thus, one can move iteratively, searching for critical points within hyperrectangles, in order to arrive at \eqref{max7}.
}

We next demonstrate how forward reachable sets are over-approximated via Proposition \ref{prop:p1} and under-approximated via Theorem \ref{thm2}.  Specifically, we take $\W = [-1/4,\, 0] \times [0,\, 1/4]$ and approximate $R^{F}(1/2;\, \X_0)$ for $\X_0 = [-1/2,\, 1/2]^3$.  An over-approximation of $R^{F}(1/2;\, \X_0)$ is computed by simulating the system \eqref{embedding}, here taken relative to $\delta$, forward in time for $T=1/2$. Additionally, note that \eqref{cs_sys} satisfies the hypothesis of Special Case \ref{spec2}, and thus \change{$\Delta(x,\, w,\, \widehat{w},\, \widehat{x}) = - \delta(\widehat{w},\, \widehat{x},\, x,\, w)$} is a tight decomposition function for the backward-time system \change{\eqref{backward}}.  An under-approximation of $R^{F}(1/2;\, \X_0)$ is computed by simulating the system \eqref{under_embedding}, here taken relative to $\Delta$, forward in time for $T=1/2$. \change{Simulation results are provided in Figure \ref{fig1}.}
\vspace{.35 cm}

\begin{figure}[t]
    \begin{subfigure}{0.49\textwidth}
        \input{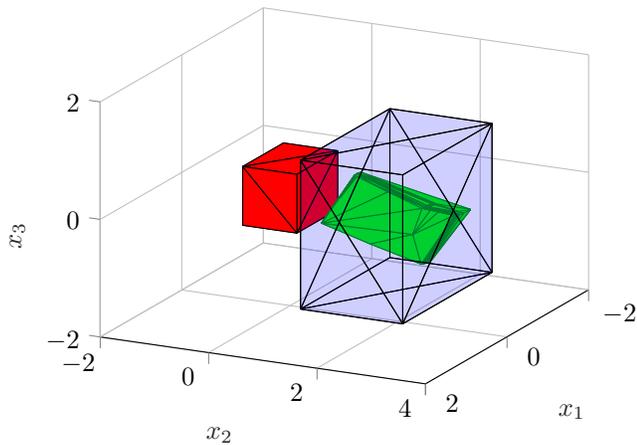}
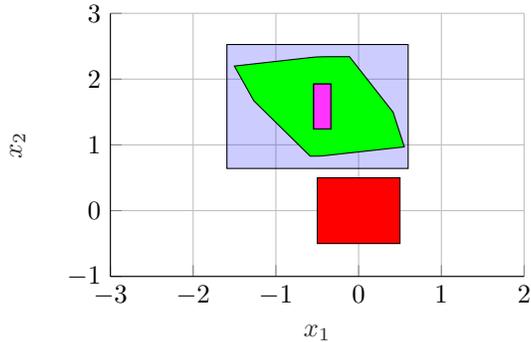
        \caption{\change{Numerical Example: Simulation results. }
        \vspace{.23cm}
        }
        \label{fig:ex1_1}
    \end{subfigure}
    ~
    \begin{subfigure}{0.49\textwidth}
%
%
\definecolor{mycolor1}{rgb}{1.00000,0.00000,1.00000}%
\begin{tikzpicture}

\begin{axis}[%
width=5.5cm,
height=3.5cm,
at={(0cm,0cm)},
scale only axis,
xmin=-3,
xmax=2,
xtick={-3, -2, -1,  0,  1,  2},
xlabel style={font=\color{white!15!black}},
xlabel={$x_1$},
ymin=-1,
ymax=3,
ytick={-1,  0,  1,  2,  3},
ylabel style={font=\color{white!15!black}},
ylabel={$x_2$},
axis background/.style={fill=white},
axis x line*=bottom,
axis y line*=left,
xmajorgrids,
ymajorgrids
]

\addplot[area legend, draw=black, fill=blue, fill opacity=0.2, forget plot]
table[row sep=crcr] {%
x	y\\
-1.59182984631072	0.641429422424022\\
0.598596587653922	0.641429422424022\\
0.598596587653922	2.52626415635019\\
-1.59182984631072	2.52626415635019\\
}--cycle;

\addplot[area legend, draw=black, fill=red, forget plot]
table[row sep=crcr] {%
x	y\\
-0.5	-0.5\\
0.5	-0.5\\
0.5	0.5\\
-0.5	0.5\\
}--cycle;

\addplot[area legend, draw=black, fill=green, forget plot]
table[row sep=crcr] {%
x	y\\
-0.586730164986126	0.829870712969223\\
-0.58665194400575	0.829804167418225\\
-0.586582726493297	0.829755358210305\\
-0.586277399872485	0.829557654727639\\
-0.586016481870774	0.829451731669768\\
-0.585378812214246	0.829262788279804\\
-0.585103435802985	0.829201383685806\\
-0.584849655102515	0.829164724806695\\
-0.584286480631046	0.829086511348727\\
-0.583191838419502	0.828936298931295\\
-0.583124994401041	0.828929720106663\\
-0.583061672224799	0.828924083824232\\
-0.582461767609049	0.828887470168227\\
-0.581480187998306	0.828834942597829\\
-0.580921006366798	0.828808765284895\\
-0.57941076969325	0.828781276858969\\
-0.578430369268246	0.828763471510794\\
-0.578180343743965	0.828761442232369\\
-0.576670118456095	0.828751314917555\\
-0.575159893168224	0.828749868158295\\
-0.574180088169588	0.828749423906616\\
-0.573622672161384	0.828749266558959\\
-0.571481985831568	0.828749266558959\\
-0.569456005494717	0.828749266558959\\
-0.569341299501753	0.828749266558959\\
-0.567315319164902	0.828749266558959\\
-0.567200613171938	0.828749266558959\\
-0.56528933882805	0.828749266558959\\
-0.56528933882805	0.828749266558959\\
-0.565174632835086	0.828749266558959\\
-0.563148652498235	0.828749266558959\\
-0.563148652498235	0.828749266558959\\
-0.563033946505271	0.828749266558959\\
-0.561122672161383	0.828749266558959\\
-0.56100796616842	0.828749266558959\\
-0.56100796616842	0.828749266558959\\
-0.558981985831568	0.828749266558959\\
-0.558981985831568	0.828749266558959\\
-0.558867279838605	0.828749266558959\\
-0.558867279838604	0.828749266558959\\
-0.556956005494717	0.828749266558959\\
-0.556956005494717	0.828749266558959\\
-0.556841299501753	0.828749266558959\\
-0.556841299501753	0.828749266558959\\
-0.554815319164902	0.828749266558959\\
-0.554700613171938	0.828749266558959\\
-0.554700613171938	0.828749266558959\\
-0.55278933882805	0.828749266558959\\
-0.552674632835086	0.828749266558959\\
-0.552674632835086	0.828749266558959\\
-0.550648652498235	0.828749266558959\\
-0.550648652498235	0.828749266558959\\
-0.550533946505271	0.828749266558959\\
-0.550533946505271	0.828749266558959\\
-0.548622672161383	0.828749266558959\\
-0.54850796616842	0.828749266558959\\
-0.546481985831568	0.828749266558959\\
-0.546367279838604	0.828749266558959\\
-0.544341299501753	0.828749266558959\\
-0.542200613171938	0.828749266558959\\
-0.533867279838604	0.828770099892293\\
-0.529700613171938	0.828780516558959\\
-0.525534467338605	0.828811766558959\\
-0.521368321505271	0.828843016558959\\
-0.517202175671938	0.828874266558959\\
-0.504708425671938	0.829061688433959\\
-0.503659814239717	0.829079442470813\\
-0.495336894968883	0.829287512783313\\
-0.491175435333467	0.829391547939563\\
-0.490914827596445	0.829398299123904\\
-0.478446053963632	0.829865854006716\\
-0.474297588667408	0.830083665666873\\
-0.470149123371184	0.830301477327029\\
-0.46600065807496	0.830519288987185\\
-0.465418268814779	0.830558844224436\\
-0.457143119388347	0.831138389641102\\
-0.453005544675131	0.831428162349436\\
-0.448882458597332	0.831799722740036\\
-0.444759372519532	0.832171283130637\\
-0.440636286441733	0.832542843521237\\
-0.439606944374541	0.832644587359944\\
-0.427293420199733	0.834033490093584\\
0.552232902850375	0.969334349144185\\
0.55220120065438	0.969566606270461\\
0.450716067634015	1.36520765434753\\
0.416878335270303	1.49707482035959\\
0.153109590902689	1.91929632160202\\
-0.110539249152692	2.34103110609616\\
-0.113704436617246	2.34106182791836\\
-0.16537390865571	2.34110089041836\\
-0.169540575322377	2.34110089041836\\
-0.173707241989044	2.34110089041836\\
-0.17787390865571	2.34110089041836\\
-0.182040575322377	2.34110089041836\\
-0.182040575322377	2.34110089041836\\
-0.182040575322377	2.34110089041836\\
-0.185349253869078	2.34110089041836\\
-0.186207241989044	2.34110089041836\\
-0.186207241989044	2.34110089041836\\
-0.189515920535745	2.34110089041836\\
-0.19037390865571	2.34110089041836\\
-0.193682587202411	2.34110089041836\\
-0.197849253869078	2.34110089041836\\
-0.197849253869078	2.34110089041836\\
-0.202015920535745	2.34110089041836\\
-0.202015920535745	2.34110089041836\\
-0.205324599082446	2.34110089041836\\
-0.206182587202411	2.34110089041836\\
-0.209491265749112	2.34110089041836\\
-0.210349253869078	2.34110089041836\\
-0.213657932415779	2.34110089041836\\
-0.213657932415779	2.34110089041836\\
-0.217824599082446	2.34110089041836\\
-0.217824599082446	2.34110089041836\\
-0.221991265749112	2.34110089041836\\
-0.225299944295814	2.34110089041836\\
-0.226157932415779	2.34110089041836\\
-0.22946661096248	2.34110089041836\\
-0.230324599082446	2.34110089041836\\
-0.233633277629147	2.34110089041836\\
-0.237799944295814	2.34110089041836\\
-0.24196661096248	2.34110089041836\\
-0.24196661096248	2.34110089041836\\
-0.246133277629147	2.34110089041836\\
-0.250299944295814	2.34110089041836\\
-0.294127278678016	2.34106945731082\\
-0.306627278678016	2.34103820731082\\
-0.343094327707978	2.3408818135859\\
-0.355591054622184	2.3407880635859\\
-0.385245384679924	2.34049918112994\\
-0.397732414923254	2.34031175925494\\
-0.421183420410624	2.33991999753136\\
-0.439096216612258	2.33948105939641\\
-0.451564112377526	2.33916895758473\\
-0.457273501694034	2.33900366309267\\
-0.470361655846763	2.33858742243861\\
-0.479372950907017	2.33826586349866\\
-0.491809406731973	2.33779832959487\\
-0.504199398444867	2.33714496461556\\
-0.516525364075096	2.33627582389413\\
-0.528767382920884	2.3351615207737\\
-0.540903368995468	2.33377333149095\\
-1.50174805862835	2.19856980456883\\
-1.26758138918208	1.67217205316997\\
-0.586730164986126	0.829870712969223\\
}--cycle;

\addplot[area legend, draw=black, fill=white, forget plot]
table[row sep=crcr] {%
x	y\\
-0.543901291435087	1.24285401053452\\
-0.335442098829544	1.24285401053452\\
-0.335442098829544	1.9259668701204\\
-0.543901291435087	1.9259668701204\\
}--cycle;

\addplot[area legend, draw=black, fill=mycolor1, fill opacity=0.8, forget plot]
table[row sep=crcr] {%
x	y\\
-0.543901291435087	1.24285401053452\\
-0.335442098829544	1.24285401053452\\
-0.335442098829544	1.9259668701204\\
-0.543901291435087	1.9259668701204\\
}--cycle;
\end{axis}
\end{tikzpicture}%
        \caption{
        \change{Projection of Figure \ref{fig:ex1_1} onto the $x_1$-$x_2$ plane.}
        \vspace{.23cm}
        }
        \label{fig:ex1_2}
    \end{subfigure}
    \caption{  
    \change{
    Approximating forward reachable sets for \eqref{cs_sys} from the set of initial conditions $\X_0 = [-1/2,\, 1/2]^{3}$, shown in red. The disturbance bound is given by  $\W = [-1/4,\, 0]\times [0,\, 1/4]$. 
    $R^{F}(1/2;\, \X_0)$ is computed via exhaustive simulation and is shown in green.
    A hyperrectangular over-approximation of $R^{F}(1/2;\, \X_0)$ is computed from the embedding system \eqref{embedding} as described in Proposition \ref{prop:p1} and is shown in light blue.  A hyperrectangular under-approximation of $R^{F}(1/2;\, \X_0)$ is computed from \eqref{under_embedding} as described in Theorem \ref{thm2} and is shown in pink. 
    }
    }
    \label{fig1}
\end{figure}

\section{Discussion and Conclusion}
A mixed-monotone system is generally mixed-monotone with respect to many decomposition functions and, as such, we can expect the system \eqref{eq1} to induce decomposition functions other than that constructed in \eqref{opt_decomp}.  However, some decomposition functions may be more/less conservative than others when used with Proposition \ref{prop:p1}, and we have shown that \eqref{opt_decomp} is the least conservative in the sense that it provides the tightest rectangular approximations of reachable sets
\change{
when used with Proposition \ref{prop:p1} and the existing analysis tool for mixed-monotone systems.
}

\change{
As demonstrated in the examples of this work, however, a closed form solution to \eqref{opt_decomp} is generally characterized piecewise and the number of pieces can scale exponentially in the dimension of the system state and disturbance spaces.
As argued in, \emph{e.g.}, \cite{coogan2015efficient}, a significant feature of mixed-monotone systems theory is that is that the computational complexity of reachable set computations scales linearly in the state dimension; this is not true when the computational complexity involved in evaluating the decomposition function scales exponentially in state and disturbance dimension.
Thus, in certain instances it may be preferable to use alternate decomposition functions to that constructed in \eqref{opt_decomp}.
}


\change{Theorem \ref{thm1} suggests a theory as to how decomposition functions should be formed in the general setting of \eqref{eq1};} in particular, we observe that a decomposition function $d$ should be \emph{large} when its first two inputs are larger than its second two inputs and \emph{small} when its first two inputs are smaller than its second two inputs. This is because $d(x,\, w,\, \widehat{x},\, \widehat{w})$ governs the movement of the first $n$ entries of $\Phi^e$ when $x\preceq \widehat{x}$ and $w\preceq \widehat{w}$ and therefore should be \emph{large} in order to attain tight approximations.  Likewise, $d(\widehat{x},\, \widehat{w}, x,\, w)$ governs the movement of the second $n$ entries of $\Phi^e$ and therefore should be \emph{small}. 
\change{
Note however that there is an intrinsic maximum/minimum evaluation of $d(x,\,w,\, \widehat{x},\, \widehat{w})$ and, as shown in Theorem \ref{thm1}, this bound is attained only if $d$ is tight.
}


\bibliography{Bibliography}
\bibliographystyle{ieeetr}
\end{document}